\documentclass[11pt]{article}

\usepackage[margin=1in]{geometry}
\usepackage{CJK, hyperref, amsmath, amsthm}

\newtheorem{lemma}{Lemma}
\newtheorem{theorem}{Theorem}
\newtheorem{corollary}{Corollary}

\theoremstyle{definition}
\newtheorem{definition}{Definition}

\DeclareMathOperator{\poly}{poly}
\DeclareMathOperator{\Span}{span}
\DeclareMathOperator{\tr}{tr}

\begin{document}

\begin{CJK*}{UTF8}{}

\title{Dynamics of Renyi entanglement entropy in local quantum circuits with charge conservation\footnote{Dedicated to the Chinese New Year of the Pig.}}

\CJKfamily{gbsn}

\author{Yichen Huang (黄溢辰)\\
Microsoft Research AI\\
Redmond, Washington 98052, USA\\
yichen.huang@microsoft.com}

\maketitle

\end{CJK*}

\begin{abstract}

In local quantum circuits with charge conservation, we initialize the system in random product states and study the dynamics of the Renyi entanglement entropy $R_\alpha$. We rigorously prove that $R_\alpha$ with Renyi index $\alpha>1$ at time $t$ is $\le O(\sqrt{t\ln t})$ if the transport of charges is diffusive. Very recent numerical results of Rakovszky et al. show that this upper bound is saturated (up to the sub-logarithmic correction) in random local quantum circuits with charge conservation.

\end{abstract}

\section{Introduction}

Entanglement, a concept of quantum information theory, has been widely used in condensed matter and statistical physics to provide insights beyond those obtained via ``conventional'' quantities. For example, in a non-integrable system the growth of the von Neumann entanglement entropy reveals the ballistic light cone set by the Lieb-Robinson bound \cite{LR72}, while the energy transport is diffusive \cite{KH13}.

The von Neumann entanglement entropy is the standard entanglement measure for pure states. However, it is also instructive to study the Renyi entanglement entropy, which reflects the entanglement spectrum \cite{LH08} and is easier to measure experimentally \cite{AD12, DPSZ12, IMP+15}. It is known that the von Neumann and Renyi entanglement entropies may behave differently in some cases: from the scaling of eigenstate entanglement \cite{DLL18, PGG18, LG17} to describing the simulability of quantum many-body states \cite{VC06, SWVC08, Hua15}.

Quantum circuits are not only a model of quantum computation \cite{Deu89, Yao93, NC10}, but also useful for the study of quantum many-body systems \cite{Vid08, CGW10, HC15}. Since the dynamics of local Hamiltonians can be simulated by local quantum circuits \cite{Llo96}, one may gain insights into the former from the latter. In particular, some recent works \cite{NRVH17, vRPS18, NVH18} studied random local quantum circuits, which are minimal models of quantum chaotic dynamics. Furthermore, the evolution governed by time-independent chaotic local Hamiltonians preserves energy. Such evolution is more faithfully described by random local quantum circuits with conserved quantities, in which the diffusive transport is easily observed \cite{KVH18, RPv18}.

In (not necessarily random) local quantum circuits with charge conservation, we initialize the system in random product states and study the dynamics of the Renyi entanglement entropy $R_\alpha$. Perhaps surprisingly, we rigorously prove that $R_\alpha$ with Renyi index $\alpha>1$ at time $t$ is $\le O(\sqrt{t\ln t})$ if the transport of charges is diffusive. It is straightforward to extend this result to cases where the transport is sub- or super-diffusive. Indeed, the proof explicitly shows that the growth of $R_\alpha$ with $\alpha>1$ is a probe of transport. This is in contrast to the linear (in $t$) growth of the von Neumann entanglement entropy.

We conjecture that the upper bound $O(\sqrt{t\ln t})$ on the Renyi entanglement entropy $R_\alpha$ with $\alpha>1$ holds more generally for the dynamics of time-independent local Hamiltonians with diffusive energy transport. It is known to fail \cite{FC08} in the integrable $XY$ chain, whose transport is ballistic.

\section{Preliminaries}

We start with some basic definitions. We use the natural logarithm throughout this paper.

\begin{definition} [entanglement entropy]
The Renyi entanglement entropy $R_\alpha$ with index $\alpha\in(0,1)\cup(1,+\infty)$ of a bipartite pure state $\rho_{AB}=|\psi\rangle\langle\psi|$ is defined as
\begin{equation}
R_\alpha(\rho_A)=\frac{1}{1-\alpha}\ln\tr\rho_A^\alpha=\frac{1}{1-\alpha}\ln\sum_{i\ge1}\Lambda_i^\alpha,
\end{equation}
where $\Lambda_1\ge\Lambda_2\ge\cdots\ge0$ with $\sum_{i\ge1}\Lambda_i=1$ are the eigenvalues (in descending order) of the reduced density matrix $\rho_A=\tr_B\rho_{AB}$. The min-entropy is defined as
\begin{equation}
R_\infty(\rho_A):=\lim_{\alpha\to+\infty}R_\alpha(\rho_A)=-\ln\Lambda_1.
\end{equation}
Note that the von Neumann entanglement entropy is given by
\begin{equation}
\lim_{\alpha\to1}R_\alpha(\rho_A)=-\tr(\rho_A\ln\rho_A).
\end{equation}
\end{definition}

\begin{lemma}
For $\alpha>1$, we have
\begin{equation}
R_\infty(\rho_A)\le R_\alpha(\rho_A)\le\frac{\alpha}{\alpha-1}R_\infty(\rho_A).
\end{equation}
\end{lemma}

\begin{proof}
The first inequality is a consequence of fact that $R_\alpha$ is monotonically non-increasing with respect to $\alpha$ (this is why $R_\infty$ is called the min-entropy). The second inequality follows from
\begin{equation}
R_\alpha(\rho_A)=\frac{1}{1-\alpha}\ln\sum_{i\ge1}\Lambda_i^\alpha\le\frac{1}{1-\alpha}\ln\Lambda_1^\alpha=\frac{\alpha}{\alpha-1}R_\infty(\rho_A).
\end{equation}
\end{proof}

\begin{definition} [local quantum circuit with charge conservation]
Consider a chain of $2n$ spin-$1/2$'s. Let the time-evolution operator be
\begin{equation}
U(t,0)=U(t,t-1)U(t-1,t-2)\cdots U(1,0),
\end{equation}
where $t$ is a positive integer. Each layer of the circuit consists of two sub-layers of local unitaries:
\begin{equation} \label{lu}
U(t,t-1)=\prod_{i=1}^{n-1} U_t^{2i,2i+1}\prod_{i=1}^n U_t^{2i-1,2i}.
\end{equation}
Each unitary $U_t^{i,i+1}$ acts on two neighboring spins $i,i+1$ and is block diagonal in the $\{|00\rangle,|01\rangle,|10\rangle,|11\rangle\}$ basis:
\begin{equation}
U_t^{i,i+1}=
\begin{pmatrix}
* & 0 & 0 & 0\\
0 & * & * & 0\\
0 & * & * & 0\\
0 & 0 & 0 & *
\end{pmatrix},
\end{equation}
i.e., $U_t^{i,i+1}$ is the direct sum of a phase factor, a unitary matrix of order $2$, and another phase factor. It should be clear that every $U_t^{i,i+1}$ and hence $U(t,0)$ preserve the total charge $\sum_{i=1}^{2n}\sigma_z^i$.
\end{definition}

Of special interest is the so-called random local quantum circuit with charge conservation, where each $U_t^{i,i+1}$ is the direct sum of a random phase factor, a Haar-random unitary matrix of order $2$, and another random phase factor. In this model, it is straightforward to prove that the transport of charges is diffusive \cite{RPv18}, i.e., the evolution of the distribution of charges $\{\langle\sigma_z^i\rangle\}_{i=1}^{2n}$ is (approximately) described by the diffusion equation. We emphasize that the main result of this paper only assumes diffusive transport and does not require randomness in $U_t^{i,i+1}$ or $U(t,0)$.

\section{Results}

We are ready to state and prove the main result of this paper.

\begin{theorem} \label{thm}
Consider the spin chain as a bipartite quantum system $A\otimes B$. Subsystem $A$ consists of spins $1,2,\ldots,n$, i.e., we study the entanglement across the middle cut. Initialize the system in a random product state $|\psi_{\rm init}\rangle$ in the $\sigma_x$ basis, i.e., each spin is in either $|+\rangle:=\frac{|0\rangle+|1\rangle}{\sqrt2}$ or $|-\rangle:=\frac{|0\rangle-|1\rangle}{\sqrt2}$ with equal probability. Let $\alpha>1$ and $\rho_A:=\tr_B(U(t,0)|\psi_{\rm init}\rangle\langle\psi_{\rm init}|U^\dag(t,0))$ be the reduced density matrix of subsystem $A$ at time $t$. If the transport of charges under the dynamics $U(t,0)$ is diffusive, then
\begin{equation} \label{maineq}
R_\alpha(\rho_A)\le\frac{\alpha}{\alpha-1}O(\sqrt{t\ln t})
\end{equation}
holds with probability $\ge1-1/p(t)$, where $p$ is a polynomial of arbitrarily high degree.
\end{theorem}

\begin{proof}
We divide the spin chain into two parts. One of them, labeled by ``in,'' consists of $2m$ spins (with indices $n-m+1,n-m+2,\ldots,n+m$) near the cut, where $m$ is some positive integer to be determined later. The other part, labeled by ``out,'' is the rest of the system. The initial state can be factored into
\begin{equation}
|\psi_{\rm init}\rangle=|\psi_{\rm init}\rangle_{\rm in}\otimes|\psi_{\rm init}\rangle_{\rm out},
\end{equation}
where $|\psi_{\rm init}\rangle_{\rm in}$ and $|\psi_{\rm init}\rangle_{\rm out}$ are random product states in the ``in'' and ``out'' parts of the system, respectively. Define
\begin{equation}
|\psi_0\rangle=|0\rangle^{\otimes 2m}_{\rm in}\otimes|\psi_{\rm init}\rangle_{\rm out}
\end{equation}
so that $|\langle\psi_0|\psi_{\rm init}\rangle|=2^{-m}$. Since $U(t,0)$ is unitary, we have
\begin{equation}
|\langle U(t,0)\psi_0,U(t,0)\psi_{\rm init}\rangle|=2^{-m}.
\end{equation}
The left-hand side of this equation is the absolute value of the inner product between $U(t,0)|\psi_0\rangle$ and $U(t,0)|\psi_{\rm init}\rangle$. Occasionally we do not use the standard Dirac notation because it is cumbersome.

Let $Z$ with $|Z|=2^{2n-2}$ be the set of all computational basis states (i.e., product states in the $\sigma_z$ basis) obeying the constraint that spins $n$ and $n+1$ are in the state $|00\rangle$. Let $P$ be the projection onto the subspace $\Span Z$. The state $|\psi_0\rangle$ has an extended region of $|0\rangle$'s in the middle of the chain. Since the transport of charges is described by the diffusion equation, we have
\begin{equation} \label{err}
\|(1-P)U(t,0)|\psi_0\rangle\|\le e^{-\Omega(m^2/t)}.
\end{equation}

Assume without loss of generality that $n$ is odd. The only local unitary in $U(t,t-1)$ acting on both subsystems $A$ and $B$ is in the second product on the right-hand of Eq. (\ref{lu}). Define a modified local quantum circuit
\begin{align}
&V(t,0)=V(t,t-1)V(t-1,t-2)\cdots V(1,0),\nonumber\\
&V(t,t-1)=\prod_{i=1}^{n-1} U_t^{2i,2i+1}\prod_{i=1}^{(n-1)/2}U_t^{2i-1,2i}u_t^{n,n+1}\prod_{i=(n+3)/2}^nU_t^{2i-1,2i},
\end{align}
where $u_t^{n,n+1}:=\langle00|U_t^{n,n+1}|00\rangle$ is a complex number. It is easy to see that
\begin{equation}
U(t,t-1)P=V(t,t-1)P.
\end{equation}
Therefore,
\begin{align}
U(t,0)|\psi_0\rangle&=U(t,t-1)U(t-1,0)|\psi_0\rangle\approx U(t,t-1)PU(t-1,0)|\psi_0\rangle\nonumber\\
&=V(t,t-1)PU(t-1,0)|\psi_0\rangle\approx V(t,t-1)U(t-1,0)|\psi_0\rangle,
\end{align}
where each approximation step generates an additive error upper bounded by the right-hand side of (\ref{err}). Iterating this process, we have
\begin{equation}
\||\Delta_t\rangle\|\le te^{-\Omega(m^2/t)},\quad|\Delta_t\rangle:=U(t,0)|\psi_0\rangle-V(t,0)|\psi_0\rangle.
\end{equation}

Recall that both $|\psi_{\rm init}\rangle_{\rm in}$ and $|\psi_{\rm init}\rangle_{\rm out}$ are random product states in the $\sigma_x$ basis. We now fix the latter but not the former. Then, $|\psi_0\rangle$ is fixed but $|\psi_{\rm init}\rangle$ is not. Let
\begin{equation}
S=\{|+\rangle,|-\rangle\}^{\otimes 2m}_{\rm in}\otimes|\psi_{\rm init}\rangle_{\rm out}
\end{equation}
be the set of all possible initial states consistent with $|\psi_{\rm init}\rangle_{\rm out}$ so that $|S|=2^{2m}$. Since the states in $S$ are pairwise orthogonal,
\begin{equation}
\sum_{|\psi_{\rm init}\rangle\in S}|\langle\Delta_t|U(t,0)|\psi_{\rm init}\rangle|^2\le\||\Delta_t\rangle\|^2\implies\frac{1}{|S|}\sum_{|\psi_{\rm init}\rangle\in S}|\langle\Delta_t|U(t,0)|\psi_{\rm init}\rangle|\le2^{-m}\||\Delta_t\rangle\|,
\end{equation}
where we used the inequality of arithmetic and geometric means. Define a subset of $S$ as
\begin{equation}
S':=\{|\psi_{\rm init}\rangle\in S:|\langle\Delta_t|U(t,0)|\psi_{\rm init}\rangle|\le2^{-m}\||\Delta_t\rangle\|p(t)\}.
\end{equation}
Markov's inequality implies that
\begin{equation}
|S'|/|S|\ge1-1/p(t).
\end{equation}
It suffices to prove (\ref{maineq}) for all states in $S'$. To this end, we make use of

\begin{lemma} [Eckart-Young theorem \cite{EY36}] \label{EY}
Let
\begin{equation}
|\psi\rangle=\sum_{i\ge1}\lambda_i|a_i\rangle_A\otimes|b_i\rangle_B
\end{equation}
be the Schmidt decomposition of the state $|\psi\rangle$, where $\lambda_1\ge\lambda_2\ge\cdots>0$ with $\sum_{i\ge1}\lambda_i^2=1$ are the Schmidt coefficients in descending order. Any state $|\phi\rangle$ of Schmidt rank $D$ satisfies
\begin{equation}
|\langle\phi|\psi\rangle|\le|\langle\psi'|\psi\rangle|=\sqrt{\sum_{i=1}^D\lambda_i^2}
\end{equation}
where
\begin{equation}
|\psi'\rangle:=\frac{1}{\sqrt{\sum_{i=1}^D\lambda_i^2}}\sum_{i=1}^D\lambda_i|a_i\rangle_A\otimes|b_i\rangle_B.
\end{equation}
\end{lemma}

For any particular state $|\psi_{\rm init}\rangle\in S'$, we have
\begin{align} \label{overlap}
&|\langle V(t,0)\psi_0,U(t,0)\psi_{\rm init}\rangle|=|\langle U(t,0)\psi_0,U(t,0)\psi_{\rm init}\rangle-\langle \Delta_t|U(t,0)|\psi_{\rm init}\rangle|\nonumber\\
&\ge2^{-m}-|\langle\Delta_t|U(t,0)|\psi_{\rm init}\rangle|\ge2^{-m}(1-\||\Delta_t\rangle\|p(t))\ge2^{-m}(1-te^{-\Omega(m^2/t)}p(t)).
\end{align}
Let $\lambda_1$ be the largest Schmidt coefficient of $U(t,0)|\psi_{\rm init}\rangle$, and $\Lambda_1=\lambda_1^2$ be the largest eigenvalue of the reduced density matrix $\rho_A=\tr_B(U(t,0)|\psi_{\rm init}\rangle\langle\psi_{\rm init}|U^\dag(t,0))$. Since none of the local unitaries in $V(t,t-1)$ or $V(t,0)$ act on both subsystems $A$ and $B$, $V(t,0)$ does not generate any entanglement so that $V(t,0)|\psi_0\rangle$ is a product state between $A$ and $B$ (i.e., a state of Schmidt rank $1$). Combining this with (\ref{overlap}) and Lemma \ref{EY}, we have
\begin{equation}
\lambda_1\ge2^{-m}(1-te^{-\Omega(m^2/t)}p(t)).
\end{equation}
Therefore,
\begin{equation}
R_\alpha(\rho_A)\le\frac{\alpha}{\alpha-1}R_\infty(\rho_A)=-\frac{\alpha}{\alpha-1}\ln\Lambda_1=-\frac{2\alpha}{\alpha-1}\ln\lambda_1.
\end{equation}
We complete the proof by choosing $m=O(\sqrt{t\ln t})$ with a sufficiently large pre-factor hided in the Big-O notation.
\end{proof}

It is straightforward to extend Theorem \ref{thm} to cases where the transport is sub- or super-diffusive.

\begin{corollary}
In the setting of Theorem \ref{thm}, suppose that the transport of charges under $U(t,0)$ has the scaling: distance $\sim t^z$ for $0<z<1$. Then,
\begin{equation}
R_\alpha(\rho_A)\le\frac{\alpha}{\alpha-1}O(t^z\poly\ln t)
\end{equation}
holds with probability $\ge1-1/p(t)$.
\end{corollary}

It is also straightforward to extend Theorem \ref{thm} to higher spatial dimensions.

\section*{Notes}

Very recently, we became aware of a relate work \cite{RPv19}, which also studied the growth of the Renyi entanglement entropy in diffusive systems. The numerical results there show that the upper bound in Theorem \ref{thm} is saturated (up to the sub-logarithmic correction) in random local quantum circuits with charge conservation.

\bibliographystyle{abbrv}
\bibliography{diffusive}

\begin{thebibliography}{10}

\bibitem{AD12}
D.~A. Abanin and E.~Demler.
\newblock Measuring entanglement entropy of a generic many-body system with a
  quantum switch.
\newblock {\em Physical Review Letters}, 109(2):020504, 2012.

\bibitem{CGW10}
X.~Chen, Z.-C. Gu, and X.-G. Wen.
\newblock Local unitary transformation, long-range quantum entanglement, wave
  function renormalization, and topological order.
\newblock {\em Physical Review B}, 82(15):155138, 2010.

\bibitem{DPSZ12}
A.~J. Daley, H.~Pichler, J.~Schachenmayer, and P.~Zoller.
\newblock Measuring entanglement growth in quench dynamics of bosons in an
  optical lattice.
\newblock {\em Physical Review Letters}, 109(2):020505, 2012.

\bibitem{Deu89}
D.~Deutsch.
\newblock Quantum computational networks.
\newblock {\em Proceedings of the Royal Society A}, 425(1868):73--90, 1989.

\bibitem{DLL18}
A.~Dymarsky, N.~Lashkari, and H.~Liu.
\newblock Subsystem eigenstate thermalization hypothesis.
\newblock {\em Physical Review E}, 97(1):012140, 2018.

\bibitem{EY36}
C.~Eckart and G.~Young.
\newblock The approximation of one matrix by another of lower rank.
\newblock {\em Psychometrika}, 1(3):211--218, 1936.

\bibitem{FC08}
M.~Fagotti and P.~Calabrese.
\newblock Evolution of entanglement entropy following a quantum quench:
  Analytic results for the {$XY$} chain in a transverse magnetic field.
\newblock {\em Physical Review A}, 78:010306(R), 2008.

\bibitem{PGG18}
J.~R. Garrison and T.~Grover.
\newblock Does a single eigenstate encode the full {H}amiltonian?
\newblock {\em Physical Review X}, 8(2):021026, 2018.

\bibitem{Hua15}
Y.~Huang.
\newblock {\em Classical simulation of quantum many-body systems}.
\newblock PhD thesis, University of California, Berkeley, 2015.

\bibitem{HC15}
Y.~Huang and X.~Chen.
\newblock Quantum circuit complexity of one-dimensional topological phases.
\newblock {\em Physical Review B}, 91(19):195143, 2015.

\bibitem{IMP+15}
R.~Islam, R.~Ma, P.~M. Preiss, M.~E. Tai, A.~Lukin, M.~Rispoli, and M.~Greiner.
\newblock Measuring entanglement entropy in a quantum many-body system.
\newblock {\em Nature}, 528(7580):77--83, 2015.

\bibitem{KVH18}
V.~Khemani, A.~Vishwanath, and D.~A. Huse.
\newblock Operator spreading and the emergence of dissipative hydrodynamics
  under unitary evolution with conservation laws.
\newblock {\em Physical Review X}, 8(3):031057, 2018.

\bibitem{KH13}
H.~Kim and D.~A. Huse.
\newblock Ballistic spreading of entanglement in a diffusive nonintegrable
  system.
\newblock {\em Physical Review Letters}, 111(12):127205, 2013.

\bibitem{LH08}
H.~Li and F.~D.~M. Haldane.
\newblock Entanglement spectrum as a generalization of entanglement entropy:
  Identification of topological order in non-abelian fractional quantum {H}all
  effect states.
\newblock {\em Physical Review Letters}, 101(1):010504, 2008.

\bibitem{LR72}
E.~H. Lieb and D.~W. Robinson.
\newblock The finite group velocity of quantum spin systems.
\newblock {\em Communications in Mathematical Physics}, 28(3):251--257, 1972.

\bibitem{Llo96}
S.~Lloyd.
\newblock Universal quantum simulators.
\newblock {\em Science}, 273(5278):1073--1078, 1996.

\bibitem{LG17}
T.-C. Lu and T.~Grover.
\newblock Renyi entropy of chaotic eigenstates.
\newblock arXiv:1709.08784, 2017.

\bibitem{NRVH17}
A.~Nahum, J.~Ruhman, S.~Vijay, and J.~Haah.
\newblock Quantum entanglement growth under random unitary dynamics.
\newblock {\em Physical Review X}, 7(3):031016, 2017.

\bibitem{NVH18}
A.~Nahum, S.~Vijay, and J.~Haah.
\newblock Operator spreading in random unitary circuits.
\newblock {\em Physical Review X}, 8(2):021014, 2018.

\bibitem{NC10}
M.~A. Nielsen and I.~L. Chuang.
\newblock {\em Quantum Computation and Quantum Information}.
\newblock Cambridge Univerisity Press, Cambridge, UK, 2010.

\bibitem{RPv18}
T.~Rakovszky, F.~Pollmann, and C.~W. von Keyserlingk.
\newblock Diffusive hydrodynamics of out-of-time-ordered correlators with
  charge conservation.
\newblock {\em Physical Review X}, 8(3):031058, 2018.

\bibitem{RPv19}
T.~Rakovszky, F.~Pollmann, and C.~W. von Keyserlingk.
\newblock Sub-ballistic growth of {R}\'enyi entropies due to diffusion.
\newblock arXiv:1901.10502, 2019.

\bibitem{SWVC08}
N.~Schuch, M.~M. Wolf, F.~Verstraete, and J.~I. Cirac.
\newblock Entropy scaling and simulability by matrix product states.
\newblock {\em Physical Review Letters}, 100(3):030504, 2008.

\bibitem{VC06}
F.~Verstraete and J.~I. Cirac.
\newblock Matrix product states represent ground states faithfully.
\newblock {\em Physical Review B}, 73(9):094423, 2006.

\bibitem{Vid08}
G.~Vidal.
\newblock Class of quantum many-body states that can be efficiently simulated.
\newblock {\em Physical Review Letters}, 101(11):110501, 2008.

\bibitem{vRPS18}
C.~W. von Keyserlingk, T.~Rakovszky, F.~Pollmann, and S.~L. Sondhi.
\newblock Operator hydrodynamics, {OTOC}s, and entanglement growth in systems
  without conservation laws.
\newblock {\em Physical Review X}, 8(2):021013, 2018.

\bibitem{Yao93}
A.~C.-C. Yao.
\newblock Quantum circuit complexity.
\newblock In {\em Proceedings of the 34th IEEE Annual Symposium on Foundations
  of Computer Science}, pages 352--361, 1993.

\end{thebibliography}

\end{document}